%% file: main_arxiv.tex
\documentclass[11pt,letter,final]{article}

\usepackage[left=1in,right=1in,top=1in,bottom=1in]{geometry}
\usepackage{booktabs} 
\usepackage[ruled]{algorithm2e} 

\SetAlFnt{\small}
\SetAlCapFnt{\small}
\SetAlCapNameFnt{\small}
\SetAlCapHSkip{0pt}
\IncMargin{-\parindent}

\usepackage{natbib}
\usepackage{graphicx}
\usepackage{amssymb}
\usepackage{authblk}

\usepackage{xspace}

\newcommand{\ie}{{i.e.,~\xspace}}

\newcommand{\eg}{{e.g.,~\xspace}}

\newcommand{\sbr}[1]{\left[\,#1\,\right]}

\newcommand{\xhdr}[1]{\vspace{1mm} \noindent{\bf #1}}

\usepackage{color-edits}
\addauthor{kh}{red}
\addauthor{as}{blue}
\addauthor{bl}{magenta}

\usepackage{amsmath}
\usepackage{amsfonts}
\usepackage{cleveref}

\usepackage{amsthm}
\usepackage{thmtools}
\usepackage{thm-restate}

\newcommand{\cK}{\mathcal{K}}
\newcommand{\cT}{\mathcal{T}}
\newcommand{\cR}{\mathcal{R}}

\newcommand{\cM}{\mathcal{M}}

\newcommand{\cA}{\mathcal{A}}

\newcommand{\cC}{\mathcal{C}}
\newcommand{\cD}{\mathcal{D}}

\newcommand{\cP}{\mathcal{P}}
\newcommand{\cQ}{\mathcal{Q}}

\newcommand{\vp}{\mathbf{p}}
\newcommand{\vx}{\mathbf{x}}

\newcommand{\lp}{\mathrm{LP}}
\newcommand{\scdp}{\texttt{set\_cover}}

\newcommand{\nadp}{\texttt{non-adaptive}}

\newtheorem{theorem}{Theorem}[section]
\newtheorem{lemma}[theorem]{Lemma}

\newtheorem{remark}[theorem]{Remark}
\newtheorem{corollary}[theorem]{Corollary}

\newtheorem{assumption}[theorem]{Assumption}
\newtheorem{example}[theorem]{Example}
\newtheorem{definition}[theorem]{Definition}

\begin{document}

\title{Algorithmic Persuasion Through Simulation}

\author[1]{Keegan Harris\footnote{Some of the results were obtained while the author was an intern at Microsoft Research.}}
\author[2]{Nicole Immorlica}
\author[2]{Brendan Lucier}
\author[2]{Aleksandrs Slivkins}

\affil[1]{Carnegie Mellon University, \texttt{keeganh@cs.cmu.edu}}
\affil[2]{Microsoft Research, \texttt{\{nicimm,brlucier,slivkins\}@microsoft.com}}

\date{First version: November 2023\\ This version: February 2025}

\maketitle

\begin{abstract}
\input{paper/main_body/abs.tex}
\end{abstract}

\input{paper/main_body/intro-v4}
\input{paper/main_body/related}
\input{paper/main_body/model}

\input{paper/main_body/preliminaries}
\input{paper/main_body/query_computation}
\input{paper/main_body/extensions}

\input{paper/main_body/non_binary}
\input{paper/main_body/conc} 

\bibliographystyle{ACM-Reference-Format}
\bibliography{bibliographies/bib-abbrv, bibliographies/bib-AGT, bibliographies/bib-ML, bibliographies/bib-bandits, bibliographies/bib-slivkins, bibliographies/refs}


\appendix
\input{paper/appendix/extension-app}


\end{document}

%% file: paper/main_body/abs.tex
We study a Bayesian persuasion game where a sender wants to persuade a receiver to take a binary action, such as purchasing a product. 
The sender is informed about the (real-valued) state of the world, such as the quality of the product, 
but only has limited information about the receiver's beliefs and utilities. 
Motivated by customer surveys, user studies, and recent advances in AI, we allow the sender to learn more about the receiver by querying an oracle that simulates the receiver's behavior. 
After a fixed number of queries, the sender commits to a messaging policy and the receiver takes the action that maximizes her expected utility given the message she receives. 
We characterize the sender's optimal messaging policy given any distribution over receiver types. 
We then design a polynomial-time querying algorithm that optimizes the sender’s expected utility in this game. 
We also consider approximate oracles, more general query structures, and costly queries.\looseness-1 

%% file: paper/main_body/intro-v4.tex
\section{Introduction}\label{sec:intro}
Information design~\citep{BergemannMorris-survey19} 
is a branch of theoretical economics that analyzes how provision of information by an informed designer influences the strategic behavior of agents in a game. We initiate the study of information design with oracle access to agent types.  This oracle is endowed with information about the agents and can be queried by the designer in order to refine her beliefs and thus improve her decision of what information to convey to the agents. 

We focus on Bayesian persuasion~\citep{kamenica2011bayesian, kamenica2019bayesian}, a paradigmatic setting in information design. Bayesian persuasion is an information design game between two players: an informed sender (\ie designer), who observes the state of the world, and an uninformed receiver (\ie agent), who does not see the state but takes an action. The payoffs of both players depend on both the world's state and the receiver's action. The game proceeds as follows: The sender commits to a messaging policy, \ie a mechanism for revealing information to the receiver about the state of the world, before the state is realized. Once the state is realized, the sender sends a message to the receiver according to their messaging policy. Upon receiving the message, the receiver updates her belief about the state of the world, and takes an action. 

The sender's payoff-maximizing messaging policy often depends on information about the receiver such as the receiver's utility function or her belief about the state of the world. The standard setting assumes the sender has \emph{full information} about the receiver, but this may not always be the case. The line of work on ``robust" Bayesian persuasion \citep[\eg][]{dworczak2022preparing, parakhonyak2022persuasion, hu2021robust} takes the other extreme, assuming the sender must construct a messaging policy with \emph{no information} about the receiver, \ie an unquantifiable uncertainty.  Such ``robust" messaging policies are more broadly applicable than those tailored to the receiver's type,  but extract less utility for the sender. 
The model of  Bayesian persuasion with a privately informed receiver \citep[\eg][]{kamenica2011bayesian,kolotilin2017persuasion} lies between these two extremes: the sender has \emph{partial information} expressed as a second-order prior.  
Another ``in-between" model is ``online" Bayesian persuasion \cite{castiglioni2020online, castiglioni2021multi, bernasconi2023optimal}, where the sender starts out uninformed (akin to robust Bayesian persuasion), but \emph{learns} over time as it interacts with a sequence of receivers (see Section~\ref{sec:related}).

We consider a model of Bayesian Persuasion that combines \emph{partial information} on the receiver (modeled as a privately informed receiver) and sender \emph{learning} to refine this partial information (for the sake of computing a better messaging policy). We posit that the sender can consult an external information source ---  put differently, \emph{query the oracle} --- at some cost. For concreteness, we focus on a sender with either a fixed budget of oracle queries, or subject to per-query costs.



In the context of Bayesian persuasion, such settings may naturally arise when, \eg the sender is a large online marketplace who experiments on some fraction of its users -- which needs to be small to avoid disrupting overall sales. 
%
%
As another example,  a startup may test out its funding pitch on smaller venture capital firms before trying to persuade a larger firm to fund their business.
Likewise, a company may run a customer focus group 
before bringing a new product or service to market.
Again, such market research is subject to a cost-benefits tradeoff.

\xhdr{Our model and results.}
We focus on \emph{simulation oracles}, natural information sources in the context of Bayesian persuasion that simulate the receiver's response in a particular persuasion instance. The sender repeatedly queries the oracle according to some \emph{querying policy}, computes its messaging policy given the oracle's output, and messages according to this policy. The querying policy may be adaptive, \ie the next query may depend on the outcomes of the previous queries. Our goal is to design {algorithms to compute (i) the querying policy and (ii)} the messaging policy given the oracle's output, so as to (approximately) optimize sender's Bayesian-expected utility. Equivalently, our goal is to compute an (approximately) optimal Perfect Bayesian Equilibrium in the ``meta-game" between sender and receiver that includes the querying policy.

We focus on Bayesian persuasion with real-valued states and receiver types and binary receiver actions. 
%
This setting 
captures scenarios in which the receiver has a single choice to make (\eg buy/don't buy, accept/reject) and there is a total ordering over states of the world (\eg product quality) and likewise over receiver types (\eg levels of skepticism). 

{Our algorithms achieve multiplicative $(1-\epsilon)$-approximation (for sender's Bayesian-expected utility) in running time polynomial in the input size and $1/\epsilon$ (and no additional dependency of the query budget).\footnote{In the language of approximation algorithms, a ``fully polynomial-time approximation scheme" (FPTAS).}
The input size in our problem is dominated by the explicit representation of the Bayesian prior as a states $\times$ types table. We also obtain a polynomial-in-$1/\epsilon$ running time (even for infinitely many states or types) after certain rounding operations on the prior are implemented. }


%
%
Our technical analysis is mainly concerned with optimizing the
querying policy. 
%
We first characterize the sender's optimal messaging policy, given arbitrary belief 
about the receiver's type (\Cref{sec:optimal.policy}).  
Thus, more fine-grained beliefs are now directly connected to higher expected sender utility.
%
%
To optimize the 
querying policy (\Cref{sec:opt_query_policy}), we design a rounding scheme over receiver types and states such that optimizing the sender's querying policy over this rounded joint distribution leads to only a small loss in expected utility. 
We then show how to compute the optimal \emph{non-adaptive} querying policy
over the rounded distribution.%
\footnote{A non-adaptive querying policy is one which specifies all queries up front.}
The final ingredient is a general ``reduction'' from adaptive to non-adaptive querying policies.
%

We  also highlight the difficulties of going beyond our main model (see~\Cref{sec:general}).
In particular, we consider more general 
query types, and prove that the resulting problem formulation is NP-hard. Moreover, we explain how the problem gets more challenging for more general state/action spaces.

%% file: paper/main_body/related.tex
\subsection{Additional Related Work}\label{sec:related}
\xhdr{Bayesian Persuasion (BP)}
was introduced by~\cite{kamenica2011bayesian}, and has been extensively studied since then, see~\cite{kamenica2019bayesian} for a recent survey.
Our model of BP with real-valued states and receiver types and binary actions is especially related to that of~\citet{kolotilin2017persuasion}, who focus on linear model of receiver utility and partial sender-receiver alignment, characterizing optimal message policies.

Robust Bayesian persuasion 
aims to relax the assumptions on the information the sender has about the receiver \citep{dworczak2022preparing, hu2021robust, parakhonyak2022persuasion, kosterina2022persuasion, zu2021learning}.
In particular, \cite{dworczak2022preparing} and \cite{hu2021robust} assume the receiver has an exogenous information source that the sender is uncertain about.
\cite{parakhonyak2022persuasion} posit that the sender is uncertain about the true distribution over states, the receiver's beliefs, and the receiver's utility function.
\cite{kosterina2022persuasion} also assume that the receiver's prior belief is unknown to the sender.
Finally,~\cite{zu2021learning} posit that neither sender nor the receiver(s) know the distribution of the payoff-relevant state.
This line of work typically focuses on characterizing the ``minimax'' messaging policy
(\ie one that is worst-case optimal over the sender's uncertainty), while our focus is on using oracle queries to help the sender overcome her uncertainty. 
Our setting is somewhat similar to \emph{online Bayesian persuasion} \cite{castiglioni2020online, castiglioni2021multi, bernasconi2023optimal}, where the sender interacts with a sequence of receivers.
%
%
In this setting, one usually assumes an \emph{adversarially chosen} sequence of receivers and minimizes sender's regret: cumulative over rounds and worst-case over receiver types. In contrast, in our setting all ``rounds" (\ie oracle queries) refer to a single receiver whose type is chosen from a known prior, and the goal is to optimize sender's Bayesian-expected utility in one Bayesian persuasion game after all the queries. Regret guarantees in online Bayesian persuasion tend to have a square-root dependence on the number of rounds and the number of types, whereas our approximation guarantees do not depend on either. (In particular, our approximation factor is entirely determined by computational considerations.) Finally, prior work on online Bayesian persuasion tends to be computationally inefficient: exponential in terms of the number of states and types.

A recent paper of \citet{lin2024information}, which appeared after (and is follow-up work with respect to) our initial working paper, studies a setting very similar to the online variant of ours.
\footnote{Their paper and ours have  been published on \texttt{arXiv.org} in October 2024 and November 2023, respectively.}
One notable difference in results is that their regret guarantee depends on the number of rounds, whereas our approximation guarantee does not depend on the number of queries. 

%
 
%
%

%
\xhdr{Pure Exploration} Our work is also conceptually similar to the line of work on \emph{pure exploration} in various bandit settings, where the first $K$ rounds of a bandit problem are used to explore the arms \citep{Tsitsiklis-bandits-04,EvenDar-icml06,Bubeck-alt09,Audibert-colt10}.
Hence, one could interpret our setting as a pure-exploration analog of online Bayesian persuasion (with a persistent receiver type drawn from a prior).

One could, in principle, apply a standard reduction from regret minimization to pure exploration using an algorithm from prior work on online Bayesian persuasion (i.e., run the algorithm, then select a messaging policy uniformly-at-random from the history). While this would yield a valid solution for our problem, this solution would suffer from the drawbacks noted above.\footnote{That is, it would be a computationally inefficient solution with approximation guarantees that  are worst-case over types (rather than Bayesian) and have square-root dependence on \#types and \#queries (as opposed to no dependence).}

\xhdr{Simulation in Other Games}
\cite{kovarik2023game} study a normal-form game setting in which one player can simulate the behavior of the other.
%
%
In contrast, we study simulation in Bayesian persuasion games, which are a type of \emph{Stackelberg} game~\citep{von2010market, conitzer2006computing}.
%
%
There is a line of work on learning the optimal strategy to commit to in Stackelberg games from query access~\citep{letchford2009learning, peng2019learning, blum2014learning, balcan2015commitment}.
However, the type of Stackelberg game considered in this line of work is different from ours.
In this setting, the leader (the leader is analogous to the sender in our setting) specifies a mixed strategy over a finite set of actions.
In contrast, in our setting the sender commits to a messaging policy which specifies a probability distribution over actions for every possible state realization.

%

Finally, some recent work on online learning in Stackelberg games \citep[\eg][]{bernasconi2023optimal, balcan2025nearly} has variants that apply to online Bayesian persuasion, albeit with the same limitations as discussed above.


%% file: paper/main_body/model.tex
\section{Setting and Background}\label{sec:setting}
We study a Bayesian persuasion game between a sender and an informed receiver (steps 2-5 in Figure~\ref{fig:our-model}).
There is a real-valued state of the world $\omega\in\Omega \subseteq [0,1]$,
revealed to the sender but not to the receiver. 
While the receiver does not see the state $\omega$, they do observe a \emph{private signal} $s \in \cT \subseteq [0,1]$, not visible to the sender.
%
%
Formally, the (state, type) pair $(\omega, s)$ is drawn from a joint distribution $F$ over $\Omega \times \cT$, called the \emph{joint prior}.
%
It is known to both the sender and the receiver.  
We associate the realized signal $s$ with the receiver's \emph{type}, and refer to $\cT$ as the \emph{type space}.
We use $\cP$ to denote the marginal distribution over types.
%
%
Crucially, $s$ can be correlated with $\omega$, which is why we refer to the receiver as \emph{informed}.
%
%
For ease of exposition we will assume for now that $\Omega$ and $\cT$ are finite sets, but we will relax that assumption in Appendix~\ref{sec:continuous}.

After receiving a message from the sender (described below), the receiver takes a binary action $a \in \cA = \{0,1\}$.
We will tend to describe action $1$ as ``taking the action'' and action $0$ as ``not taking the action."\footnote{We discuss an extension beyond binary actions in~\Cref{sec:non-binary}.} 
The sender and receiver's utilities, described by utility functions $u_S, u_R : \Omega \times \cP \times \cA \rightarrow \mathbb{R}$, are then uniquely determined by the state $\omega$, the receiver's type $s$, and the receiver's action $a$.
%
We assume the sender's utility is given by $u_S(\omega,s,a) = a$, meaning that the sender always prefers that the receiver take the action regardless of the state of the world and the receiver's type.
The receiver's utility is given by $u_R(\omega,s,a) = a \times u(\omega,s)$.  That is, the receiver always obtains utility $0$ when not taking the action, and utility $u(\omega,s)$ when taking the action.  We assume that $u(\omega,s)$ is weakly increasing in $\omega$, meaning that higher states of the world correspond to weakly higher benefits to the receiver for taking the action.

We make the following two further assumptions about the receiver's type:
\begin{assumption}\label{ass:type1}
    Receiver utility $u(\omega,s)$ is weakly increasing in $s$.  In other words, agents of higher type have weakly higher value for taking the action, for any given state of the world.
\end{assumption}
\begin{assumption}\label{ass:type2}
    The joint prior $F$ over pairs $(\omega,s)$ satisfies an increasing differences property: for any $s_1, s_2 \in \cT$ with $s_1 < s_2$ and any $\omega_1, \omega_2 \in \Omega$ with $\omega_1 < \omega_2$, $$\Pr[\omega_2 | s_2] - \Pr[\omega_1 | s_2] \geq \Pr[\omega_2 | s_1] - \Pr[\omega_1 | s_1].$$  In particular,  higher signal realizations are positively correlated with higher states of the world.
\end{assumption}
Informally, these assumptions imply that higher-type receivers are more optimistic that the state of the world is high, and have higher expected utility for taking the action given any fixed belief about the state of the world.
%

We note that receiver type plays two roles in this model: it influences receiver belief through its correlation with the state of the world, and it directly influences receiver utility.  Two special cases of the model are worth emphasizing.  First, if $F$ is a product distribution, then $\omega$ and $s$ are uncorrelated and the signal realization has no influence on the receiver's belief about $\omega$.  In this case, all receiver types share the same prior over $\omega$, but differ in their utility functions.  Second, if the receiver utility function $u(\omega,s)$ is independent of $s$, then all receiver types share the same utility function but may differ in their beliefs about the distribution of $\omega$.

%
%
%
%
%

The sender sends a message $m=\sigma(\omega)\in\cM$ to the receiver, according to some randomized \emph{messaging policy} $\sigma: \Omega \rightarrow \cM$. 
Importantly, the sender announces $\sigma$ before seeing the state. 
The receiver then takes an action $a \in \cA = \{0, 1\}$.
The sender and receiver utilities are then realized by $(\omega,s,a)$ according to utility functions $u_S,\, u_R$, respectively. 
Thus, the receiver chooses an action to maximize its Bayesian-expected utility given the joint prior $F$, her type, the messaging policy, and the realized message. We assume that the receiver breaks ties in favor of taking the action $a=1$.

\paragraph{Querying the Oracle}
The sender in our model has access to a \emph{simulation oracle} that simulates the receiver's response given the \emph{realized} private signal $s$. 
In reality, a simulation oracle may be a metaphor for running experiments on users or predicting the receiver's behavior using a machine learning model. 
Indirectly, this oracle yields information on $s$, helping the sender improve her messaging policy. Below we explain how the oracle fits into our Bayesian persuasion model.\looseness-1
\begin{definition}[Simulation Oracle]
\label{def:sim}
A simulation oracle inputs a query $q=(\sigma_q, m_q)\in\mathcal{Q}$, 
where $\sigma_q$ is a messaging policy, $m_q \in \cA$ is a message, and $\cQ$ is the set of all possible simulation queries.
    The oracle returns the receiver's  best-response given the joint prior $F$ and the realized private signal $s$, \ie $a_{q} = \arg\max_{a \in \cA} \mathbb{E}_{\omega}\sbr{u_R(\omega, s, a) \mid m_q,s}$.
    %
    %
\end{definition}
%
%
We assume the sender makes $\leq K$ oracle queries before choosing the messaging policy.
The sender makes queries adaptively, choosing each query based on the responses to the previous ones.  Formally, a \emph{query history} $H\in \mathcal{H}$ is a finite (possibly empty) sequence of query-action pairs. The sender follows some \emph{querying policy}, a function $\pi \colon \mathcal{H} \to \mathcal{Q}$ that maps query history to a next query. The sender has a \emph{messaging policy rule} that maps the final query history $H$ to the messaging policy $\sigma_H$. The full game protocol is summarized in Figure~\ref{fig:our-model}.\looseness-1

\begin{remark}
While the simulation oracle returns \emph{exact} best response, as per \Cref{def:sim}, we also consider imperfect simulation oracles, see in~\Cref{sec.approx}. Moreover, we study a natural problem variant where the sender can make unlimited queries, but each query comes at a cost;  our results readily extend to this variant, see~\Cref{app:costly.pbe}.
\end{remark}

The query history $H$ has no dependency on $\omega$ after conditioning on type $s$.  Since the messaging policy $\sigma^*_H$ is observable by the receiver, history $H$ has no further bearing on the receiver's beliefs, utility, or choice of action. In particular, it does not matter to the receiver whether the querying policy and messaging policy rule are revealed.

The designer's problem here is to compute the querying policy and the messaging policy so as to (approximately) optimize sender's Bayesian-expected reward from the corresponding Bayesian persuasion game. More specifically, we need to specify two algorithms for the sender: an algorithm to compute the querying policy (which produces the querying policy which, after being executed, produces history $H$) and another algorithm to compute the messaging policy given $H$. 


\xhdr{Notation.} 
When $|\cT| < \infty$, we use the shorthand $T := |\mathcal{T}|$. 
Given messaging policy $\sigma_H$, we define $\sigma_H(m|\omega) := \Pr\sbr{\sigma_H(\omega)=m}$, and use $m \sim \sigma_H(\omega)$ to denote a message sampled from $\sigma_H$ when the state is $\omega \in \Omega$. 
Abusing notation, we let 
$\pi(s)$
be the final query history generated by querying policy $\pi$ when the receiver's type is $s \in \cT$.


\begin{figure}[t]
   \centering
    \noindent \fbox{\parbox{\columnwidth}{
    \textbf{Bayesian Persuasion with Simulation Queries}
    \begin{enumerate}
    \setlength{\itemsep}{0pt}
    \setlength{\parsep}{0pt}
    \setlength{\topsep }{0pt}
    \item Sender uses querying policy $\pi$ to query the oracle up to $K$ times, resulting in query history $H$;\looseness-1
    \item Sender commits to a messaging policy $\sigma_H$, which is visible to the receiver;
    \item State $\omega$ is revealed privately to the sender;
    \item Message $m \sim \sigma_H(\omega)$ is sent to the receiver;
    \item Receiver chooses an action 
    $a = a(m,s)\in\cA$.
    \end{enumerate}
    }}
    \caption{Interaction protocol for our setting.}
    \label{fig:our-model}
\end{figure}


%

\xhdr{Meta-game.} 
A useful alternative characterization of our setting is
as a ``meta-game" between the sender and receiver, where the sender's strategy consists of their querying policy $\pi$ and messaging policy rule $\sigma$, and the receiver's strategy is their \emph{action rule} 
$a \colon \mathcal{M} \times \cT \to \mathcal{A}$ mapping a message and type to an action. The strategy profile is denoted $(\pi, \sigma, a)$.
%
The players also have beliefs, generated according to \emph{belief rules}.
The sender's belief rule $B_S \colon \mathcal{H} \to \Delta(\cT)$ maps query histories to distributions over receiver types.
The receiver's belief rule $B_R \colon \mathcal{M} \times \cT \to \Delta(\Omega)$ maps messages and signals to distributions over the state $\omega$.
Our objective is to compute an (approximate) Perfect Bayesian equilibrium of this meta-game.


\begin{definition}[Perfect Bayesian Equilibrium]
A strategy profile $(\pi^*, \sigma^*, a^*)$ and belief rules $B_S$, $B_R$
form a Perfect Bayesian equilibrium if
\end{definition}
\begin{enumerate}
    \setlength{\itemsep}{0pt}
    \setlength{\parsep}{0pt}
    \setlength{\topsep }{0pt}

    \item For each $m \in \mathcal{M}$ and $s \in \cT$, action $a^*(m, s)$ maximizes the receiver's expected utility given belief $B_R(m, s)$, i.e.
    $a^*(m, s) \in \arg\max_a \,
    \mathbb{E}_{\omega \sim B_R(m, s)}\sbr{ u_R( \omega, s, a) }$.
    \item Belief $B_R(m,s)$ is the correct posterior over $\omega$ given $s$, $\sigma^*_H$, and the fact that $\sigma^*_H(\omega) = m$.\footnote{It is possible that some pairs $(m,s)$ may have probability $0$, in which case $B_R(m,s)$ can be arbitrary.  We note that since $m$ is generated according to $\sigma^*_H$, which is known to the receiver, pairs $(m,s)$ of probability $0$ cannot occur even off the equilibrium path of play.\looseness-1}
    \item For each $H \in \mathcal{H}$, messaging policy $\sigma^*_H$ maximizes the sender's expected utility given belief $B_S(H)$, i.e.
    $\sigma^*_H \in \arg\max_{\sigma}\, 
    \mathbb{E}_{s \sim B_S(H),\; \omega \sim F|_s}\sbr{ u_S( \omega, a^*(\sigma(\omega), s)) }$.
    \item Belief $B_S(H)$ is a correct posterior over $s$, given $\pi^*$ and that $\pi^*$ generates history $H$.
    \item Sender's querying policy $\pi^*$ maximizes the sender's expected utility given $\sigma^*$ and $a^*$, i.e.
    $\pi^* \in \arg\max_{\pi}\;
    \mathbb{E}_{(\omega,s) \sim F,\; H \sim \pi }\sbr{ u_S( \omega, a^*(\sigma^*_H(\omega), s) ) } $.
\end{enumerate}
%

Optimality of the querying policy $\pi^*$ implies that, given any \emph{partial} history $H$ of $< K$ queries, the subsequently chosen query $\pi^*(H)$ must also be utility-optimizing for the sender given the posterior over signal $s$ induced by $H$.
Also, since the generation of histories $H$ is mechanical (given the choice of $\pi^*$ and the realization of $(s, u_R)$), any history $H$ that is inconsistent with any realization of $s$ will have probability $0$ of being observed, even off the equilibrium path of play.

\paragraph{Discussion.}
We note that we assumed the receiver has a private signal correlated with the state and thus knows something about the state that the sender does not at the beginning of the game. 
Such a setting is sometimes motivated by scenarios where the receiver has access to news that the sender does not have access to when designing her messaging policy. 
It can be interesting to consider what happens if the receiver has ``fake'' news but acts as if her news is true. 
One way to model this is that two state-signal pairs are drawn independently: $(\omega,s), (\omega',s') \sim F$, where $\omega$ is the true state but the receiver observes signal $s'$ and (incorrectly) infers that $\omega$ is distributed as $F|_{s'}$.  All of our results carry through to this setting with minor modifications to the algorithms.\looseness-1

%% file: paper/main_body/preliminaries.tex
\subsection{Preliminaries}\label{sec:think}



Our setting exhibits the following structure.  
Since the action space is binary, for any messaging policy and any message sent according to that policy, the receiver types are partitioned into two sets: those that induce action $a=1$ in response to the message and those that induce action $a=0$. 
Furthermore, 
since $\cT \subseteq [0,1]$, the heterogenity among possible receiver types is single-dimensional and totally-ordered by real value. 

\begin{lemma}
\label{lem:monotone}
    Fix any PBE $(\pi^*,\sigma^*,a^*)$.  Then for any message $m$ in the support of $\sigma^*$ and any $s_1, s_2 \in \cT$ with $s_1 < s_2$, we have $a^*(m,s_1) \leq a^*(m,s_2)$.
\end{lemma}
\begin{proof}
    The result is immediate if $a^*(m,s_1) = 0$, so assume that $a^*(m,s_1) = 1$.  Our aim is to show that $a^*(m,s_2) = 1$.  Since $\mathbb{E}_{\omega \sim B_R(m, s_2)}[u_R(\omega, s_2, 0)] = 0$, it suffices to show that $$\mathbb{E}_{\omega \sim B_R(m, s_2)}[u_R(\omega, s_2, 1)] \geq 0.$$
    %
    %
    This follows from the fact that 
    $$\mathbb{E}_{\omega \sim B_R(m, s_2)}[u(\omega, s_2)] \geq \mathbb{E}_{\omega \sim B_R(m, s_1)}[u(\omega, s_2)] \geq \mathbb{E}_{\omega \sim B_R(m, s_1)}[u(\omega, s_1)] \geq 0,$$
    where the first inequality follows from Assumption~\ref{ass:type1}, the second from Assumption~\ref{ass:type2}, and the third from the fact that $a^*(m,s_1) = 1$ and the definition of PBE.
\end{proof}

%

This result implies that there always exists a simulation query $q$ that can distinguish between any two receivers who do not behave identically on every message.
%
Furthermore, any simulation query implies a partition on receiver types defined by a threshold.\footnote{Interestingly this is not true in more general settings, as we will discuss in~\Cref{sec:non-binary}.\looseness-1} 
We will use this fact in Section~\ref{sec:equilcomp} to argue that the optimal querying policy of the sender can be characterized as choosing a set of thresholds.

%% file: paper/main_body/query_computation.tex
\section{Equilibrium Computation}
\label{sec:equilcomp}

We now develop an algorithm that computes an $O(\epsilon)$-sender-optimal equilibrium of the meta-game in time $O(\mathrm{poly}(1/\epsilon))$. 
Our approximation is additive: the sender's payoff under the calculated policies will be at least the optimal payoff minus $O(\epsilon)$. 
Our approach relies on three steps and is summarized in~\Cref{alg:total}.  
First, we discretize the state space and receiver type space to at most $O(1/\epsilon)$ entries each, by 
rounding the 
joint distribution 
appropriately. 
%
(This step may be skipped if the number of states and/or types is small.)
Using this rounded type distribution, we compute an $O(\epsilon)$-optimal \emph{non-adaptive} querying policy for $\min\{\frac{1}{\epsilon}, 2^K - 1\}$ queries, i.e. a querying policy which would be $O(\epsilon)$-optimal if the sender had to make $\min\{\frac{1}{\epsilon}, 2^K - 1\}$ queries \emph{up front}. 
Finally, we show a reduction from optimal non-adaptive querying policies to optimal adaptive querying policies to get an $O(\epsilon)$-optimal adaptive querying policy for $K$ queries. 

Our methods for computing an approximately optimal querying policy crucially make use of the sender's optimal messaging policy for a discrete set of states and receiver types. 
As such, we first characterize the sender's optimal messaging policy for this setting and show how it may be computed in polynomial time. 
Fixing the actions of the sender, the receiver's best response is simply a Bayesian update and hence is computable in constant time. 

We conclude this section by discussing settings where the oracle is an imperfect representation of the receiver's type, as is the case when, e.g. the sender is using an imperfect machine learning model to simulate the receiver's response. 

\subsection{Optimal Messaging Policy for Finite Types and States of the World}
\label{sec:optimal.policy}
%
We begin by characterizing the sender's optimal messaging policy whenever she has uncertainty about the receiver's type.
%
The optimal messaging policy has at most $|\Omega|$ messages, each with the following interpretation:
There is a message $m^*$ with a threshold receiver type 
$s^*$ 
such that a receiver with type $s^*$
is indifferent between action and inaction.
All receivers with type $s$
such that $s > s^*$
should take action $a=1$ upon receiving message $m^*$, and those for which $s < s^*$
should take action $a = 0$.
%


\begin{restatable}{proposition}{optsp}[Optimal Messaging Policy]\label{prop:opt-binary}
%
For a given (finite) set of $T$ receiver types 
and finite set $\Omega$ of $W$ possible states of the world, the sender's optimal messaging policy can be computed in time $\mathcal{O}((WT)^{2.5})$ and has non-zero mass on at most $W$ messages.\looseness-1
\end{restatable}
\begin{proof}[Proof Sketch]
    The proof proceeds by applying a revelation principle-style argument to write the problem of computing the optimal messaging policy as a linear program over $W(T+1)$ variables, which correspond to the probabilities that messages $m_0, m_1, \dotsc, m_T$ are sent in each state. 
    These messages have the property that message $m_i$ induces receivers with the 
    lowest $i$ types
    to take action $a=1$. 
    We then further simplify the program by leveraging the total ordering over receiver types and the Rank Lemma to show that the optimal messaging policy uses at most $W$ messages.\looseness-1 
\end{proof}
We demonstrate the optimal messaging policy on an example in \Cref{fig:bin-dp}.  In this example there are two states of the world, $\Omega = \{0,1\}$, and five types, each with utility function $u_R(\omega, a) = \mathbf{1}\{\omega = a\}$ and posterior beliefs that assign probabilities $(p_1,p_2,p_3,p_4,p_5) = (0.5,0.4,0.3,0.2,0.1)$ to the event that $\omega=1$, with $(\cP(p_1),\cP(p_2),\cP(p_3),\cP(p_4),\cP(p_5)) = (0.2, 0.01, 0.39, 0.2, 0.2)$.
The blue solid line represents sender's utility as a function of the cutoff index when they make no queries. Note this is not monotone. As the sender targets a higher belief $p_i$, the total mass of targeted receiver beliefs $p\geq p_i$ decreases, lowering the sender's utility. However, the probability the messaging policy can induce the receiver to take action $a=1$, conditional on the belief exceeding the target $p$, increases, improving the sender utility.  This means the sender's optimal utility might be achieved by an intermediate target (as indicated by the blue dashed line in the figure), and further complicates the problem of identifying the optimal querying policy, which we address next.\looseness-1

\begin{figure}
    \centering
    \includegraphics[width=\columnwidth]{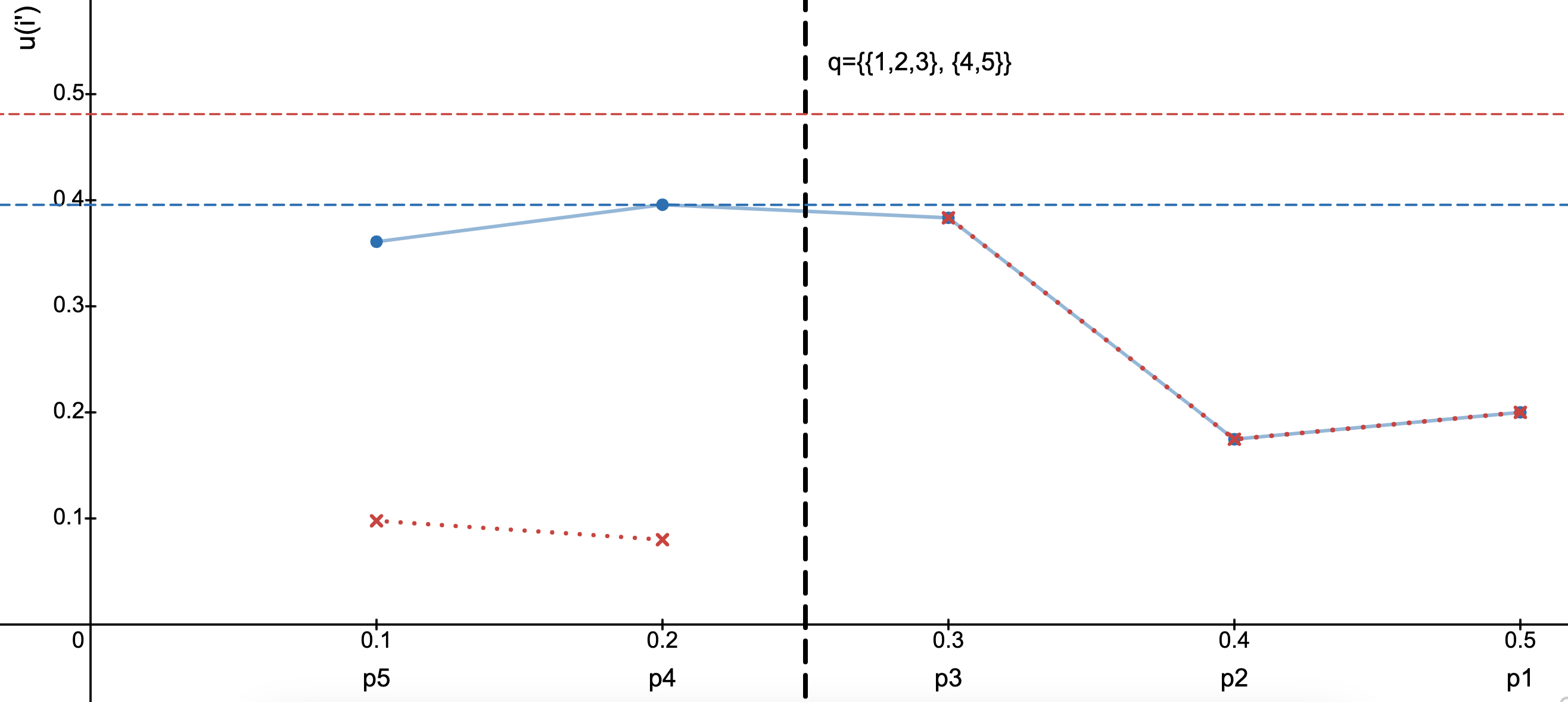}
    \caption{Sender's utility $u(i)$ as a function of cutoff belief $p_{i}$. The blue solid line is the sender's utility as a function of the cutoff index when they make no queries. The sender's optimal utility is given by the blue dashed line. The red dotted line represents sender's utility as a function of cutoff index when they make the simulation query $q$ which separates the three highest beliefs from the two lowest beliefs. The red dashed line denotes the sender's ex-ante utility from messaging optimally after making query $q$.\looseness-1
    }
    \label{fig:bin-dp}
\end{figure}

\subsection{$O(\epsilon)$-Optimal Querying Policy}
\label{sec:opt_query_policy}
%
%
%
Queries refine the sender's information about the receiver's type, so that the sender can better target her messaging policy.
\Cref{fig:bin-dp} demonstrates the impact of a single sender query separating five types (the three highest types from the two lowest ones). 
(In~\Cref{fig:bin-dp} all receivers have the same utility, and so types are equivalent to beliefs.)
The red dotted line represents the sender's utility as a function of the threshold belief in each resulting information set. This can be weakly less than the sender's utility before making the query for each individual threshold. 
However, since the sender's targeting ability improves, she is able to extract extra utility from low types, improving her overall expected utility ex ante, as represented by the red dashed line.

\xhdr{Bounding the Rounding Error for States of the World.}  Proposition~\ref{prop:opt-binary} requires that $\Omega$ be finite and introduces a runtime dependency on $|\Omega|$.  Our first step will be to discretize $\Omega$ to a small number of potential states, and have the sender's messaging policy only depend on the state rounded to the nearest point under this discretization.  To do so, we must bound the impact of such a rounding procedure on the value of the sender's optimal policy.  To this end, let $G$ denote some arbitrary distribution over $\Omega \times \cT$, with marginal distribution $G|_{\Omega}$ over states of the world.  For a fixed $\epsilon > 0$, we will define a discretization $\Omega^{(\epsilon)} \subseteq \Omega$ as follows.
For each $0 \leq k \leq 1/\epsilon$, let $\omega_k$ denote the value at quantile $k\epsilon$ of $G|_{\Omega}$.\footnote{That is, if we write $H$ for the CDF of distribution $G|_{\Omega}$, then $\omega_k$ is the unique value satisfying $H(\omega) < k\epsilon$ for all $\omega < \omega_k$ and $H(\omega) > k\epsilon$ for all $\omega > \omega_k$.  Note that we may have $\omega_k = \omega_{k+1}$ for some $k$, and that we take $\omega_0 = \min\Omega$ and $\omega_{1/\epsilon}=\max\Omega$ by convention.}  Then $\Omega^{(\epsilon)} = \{\omega_0, \omega_1, \dotsc, \omega_{1/\epsilon}\}$.

We claim that applying the above discretization to states of the world has a small impact on the sender's maximum obtainable value.  To that end, define $V^*(G)$ to be the expected utility obtained by the seller under the optimal messaging policy for distribution $G$ over $\Omega \times \cT$.  Let $V^{(\epsilon)}(G)$ be the utility of the optimal message policy that, for each $0 \leq k < 1/\epsilon$, generates the same distribution of messages for all $\omega \in [\omega_k, \omega_{k+1})$. 

\begin{restatable}{lemma}{StateDiscretization}
\label{lem:state.discretization}
For any distribution $G$ over $\Omega \times \cT$, we have $V^{(\epsilon)}(G) \leq V^*(G) \leq V^{(\epsilon)}(G) + \epsilon$.
\end{restatable}
\begin{proof}
    The first inequality follows immediately from the fact that $V^*$ maximizes over a larger space of messaging policies than $V^{(\epsilon)}$.
    
    For the second inequality, let $\sigma^*$ be an optimal messaging policy that achieves value $V^*(G)$. 
    We define a modified messaging policy $\sigma^{(\epsilon)}$ as follows.  
    Suppose the input is $\omega$ and $k$ is such that $\omega \in [\omega_k, \omega_{k+1})$.  If $k < (1/\epsilon)-1$, then messaging policy $\sigma^{(\epsilon)}$ first draws a random $\tilde{\omega}$ from $G|_\Omega$ restricted to the interval $[\omega_{k+1}, \omega_{k+2})$.  If instead $k = (1/\epsilon)-1$, then we instead draw $\sigma^{(\epsilon)}$ from $[\omega_0, \omega_1)$.  In either case, we simulate message policy $\sigma$ on input $\tilde{\omega}$ and send the resulting message.

    We note that $\sigma^{(\epsilon)}$ satisfies the condition that it assigns the same distribution of messages to all $\omega \in [\omega_k, \omega_{k+1})$ for all $0 \leq k < 1/\epsilon$.  Moreover, the marginal distribution of messages sent under $\sigma^{(\epsilon)}$ is identical to the messages sent under $\sigma^*$.  Let $E$ denote the (good) event that $\omega \leq \omega_{(1/\epsilon) - 1}$.  Conditional on event $E$, 
    we have $\tilde{\omega} \geq \omega$ with probability $1$ 
    so all receiver types that would take action $a$ under a given message from $\sigma^*$ would also take action $a$ under that same message from $\sigma^{(\epsilon)}$.

    We conclude that, conditional on the good event $E$, the sender value generated by $\sigma^{(\epsilon)}$ is at least that of $\sigma^*$.  Since the bad event where $E$ does not occur has probability at most $\epsilon$, and the sender's utility is at most $1$ on any realization, we conclude that $V^{(\epsilon)}(G) \geq V^*(G) - \epsilon$ as required.
\end{proof}

Given Lemma~\ref{lem:state.discretization}, we will assume for the remainder of this section that $|\Omega| \leq 1/\epsilon$, which comes at a loss of at most $O(\epsilon)$ in the sender's obtained expected value.


\xhdr{Bounding the Rounding Error for Types.} Our approach likewise relies on a finite type space, with a runtime dependence on $|\cT|$. 
%
We will therefore discretize $|\cT|$ to a small number of types and study the sensitivity of our policies to the resulting errors in receiver types.
To this end, given marginal distribution $\cP$ over types $\cT$ and $\epsilon > 0$, define discretization $\cT^{(\epsilon)} \subseteq \cT$ to be the set $\{s_0, s_1, \dotsc, s_{1/\epsilon}\}$, where $s_k$ is the value at quantile $k\epsilon$ of $\cP$.
Let $\cC$ be the following rounded version of the joint distribution $F$ over states of the world and receiver types: first draw $(\omega,s)$ from joint distribution $F$, then round $s$ down to the nearest element of $\cT^{(\epsilon)}$.
Given any messaging policy $\sigma$, write $V(\sigma, \cC)$ for the sender's expected payoff from messaging policy $\sigma$ when states and types are jointly distributed according to this rounded distribution $\cC$, and let $V^*(\cC) = \max_{\sigma} V(\sigma, \cC)$.  We likewise define $V(\sigma, F)$ and $V^*(F)$ for the corresponding quantities when states and types are distributed according to $F$.

We claim that, for any fixed messaging policy $\sigma$, changing the distribution of receiver types from $F$ to $\cC$ can only affect the sender's expected utility by at most $O(\epsilon)$.  In particular, this can increase or decrease the value of the optimal policy by no more than $O(\epsilon)$.

\begin{restatable}{lemma}{ValueMonotoneProb}
\label{prop:value.monotone.prob}
    Let $\cC$ be as described above.  Then for any messaging policy $\sigma$, $V(\sigma, \cC) - O(\epsilon) \leq V(\sigma, F) \leq V(\sigma, \cC) + O(\epsilon)$.
    In particular, $V^*(\cC) - O(\epsilon) \leq V^*(F) \leq V^*(\cC) + O(\epsilon)$.
\end{restatable}

\begin{proof}
    We begin with some notation. Write $\eta \colon \cT \to \cT^{(\epsilon)}$ to denote the rounding procedure that defines $\cC$.  That is, $\cC$ is the distribution of $(\omega, \eta(s))$ where $(\omega,s)$ is distributed according to $F$.  Let $\tilde{\cP}$ denote the marginal distribution of types under this rounding procedure; i.e., the distribution over $\eta(s)$ when $s$ is distributed according to $\cP$.

    Continuing with notation, 
    fix some message policy $\sigma$.  Given some quantile $q \in [0,1]$, let $s_F(q)$ denote the type at quantile $q$ from type distribution $\cP$.  Let $p_F(q)$ denote the probability that a receiver of type $s_F(q)$ takes action $a=1$, over the conditional distribution of the state $\omega$ and any randomness in the message policy $\sigma$, given that the joint distribution is $F$.  Likewise, let $s_\cC(q)$ denote the type at quantile $q$ from type distribution $\tilde{\cP}$, and let $p_{\cC}(q)$ denote the probability that a receiver of type $s_{\cC}(q)$ takes action $a=1$, over the conditional distribution of the state $\omega$ and any randomness in the message policy $\sigma$, given that the joint distribution is $\cC$.

    Note then that $V(\sigma,\cC) = \int_0^1 p_{\cC}(q) dq$ and $V(\sigma,F) = \int_0^1 p_{F}(q) dq$.  Thus, to prove the lemma, we wish to relate $p_{\cC}(q)$ and $p_{F}(q)$.  To this end, pick any $q \in [0,1]$ and choose $k$ such that $q \in [k\epsilon, (k+1)\epsilon)$. Then, from the definition of $\eta$, we have $\eta(s_{F}(q)) = s_k$ and hence $s_{\cC}(q) = s_k$.

    Now choose any message $m$ in the support of $\sigma$ such that, given joint distribution $F$, a receiver of type $s_F(q)$ takes action $a=1$ upon receiving message $m$.  Then for any $k' > k$ and any $s' \in \eta^{-1}(s_{k'})$, we have $s' \geq s_F(q)$, and hence (by Lemma~\ref{lem:monotone}) a receiver of type $s'$ would take action $a=1$ as well.  In other words, all types in $\eta^{-1}(s_{k'})$ would take action $a=1$ under joint distribution $F$.  This implies that, upon receiving message $m$ from message policy $\sigma$ when the joint distribution is $\cC$, a receiver of type $s_{k'}$ would take action $a=1$.

    By a similar argument, if under joint distribution $F$ a receiver of type $s_F(q)$ takes action $a=0$ given a certain message $m$, then for any $k' < k$ it must be that, under joint distribution $\cC$, a receiver of type $s_{k'}$ would take action $a=0$ upon receiving message $m$.

    Recall that $p_{\cC}(q) = s_k$.  From the definition of our rounding procedure, we have that $p_{\cC}(q-\epsilon) = s_{k-1}$ (if $q \geq \epsilon)$ and $p_{\cC}(q+\epsilon) = s_{k+1}$ (if $q \leq 1-\epsilon$).  Thus, 
    for any given quantile $q \in [\epsilon,1-\epsilon]$, our analysis above applied to all messages in the support of $\sigma$ implies that $p_{\cC}(q-\epsilon) \leq p_{F}(q) \leq p_{\cC}(q+\epsilon)$.

    We therefore have
    \begin{align*}
    V(\sigma,F) = \int_0^1 p_{F}(q) dq \leq \epsilon + \int_0^{1-\epsilon} p_{F}(q) dq \leq \epsilon + \int_\epsilon^1 p_{\cC}(q) dq \leq \epsilon + V(\sigma,\cC)
    \end{align*}
    and
    \begin{align*}
    V(\sigma,F) \geq \int_0^{1-\epsilon} p_{F}(q) dq \geq \int_\epsilon^1 p_{\cC}(q) dq \geq \int_0^1 p_{\cC}(q) dq -\epsilon = V(\sigma,\cC) - \epsilon
    \end{align*}
    as required.

\end{proof}

%
%
\Cref{prop:value.monotone.prob} shows that small perturbations to receiver types cannot influence the sender's payoff too much at equilibrium. 
%
Given our type distribution $\cP$ with support $\cT$,
let $\tilde{\cT}$ denote a discretized support in which each 
type $s \in \cT$
is rounded down to the nearest point in the discretization, and let $\tilde{\cP}$ denote the corresponding distribution over these rounded values (i.e. the rounded type distribution). 
Then $|\tilde{\cT}| \leq 1/\epsilon$, and by Proposition~\ref{prop:value.monotone.prob} the sender-optimal payoff under $\tilde{\cP}$ and under $\cP$ differ by at most $O(\epsilon)$. 

\xhdr{Adaptive-to-Non-Adaptive Reduction.}
We now return to the problem of optimizing over query policies.
Recall that every simulation query corresponds to a threshold over the space of types. 
Thus, given any distribution over types $\Tilde{\cP}$ with finite support of size $T$, after any (partial) history of simulation queries and responses, what is revealed to the sender is that the receiver's type $s$ lies in some interval $[\theta_L, \theta_H]$.
%
In other words, there are indices $L$ and $H$ such that the receiver's type is $s_i$ for some $L \leq i \leq H$.

Given a partial history of queries, a simple brute-force algorithm can find the myopically optimal next query in time $O(W^{2.5}T^{3.5})$.
Indeed, one can check each possible threshold query that separates types in the range $[L,H]$ (of which there are at most $H-L-1 = O(T)$ distinct options), then use the algorithm in \Cref{prop:opt-binary} to calculate the sender's optimal messaging policy given each of the two potential responses. 
This method could be used to greedily construct a sequence of queries one by one, but this may not be optimal. 
The optimal querying policy might need to make suboptimal queries in some steps to optimize the overall information at the end of the query process. 
For example, consider an instance with four possible receiver types and $K=2$ queries. 
In this scenario, it will always be optimal to use the first query to separate the smallest two receiver beliefs from the largest two beliefs, \emph{regardless of the immediate utility gain from doing so or any other parameters of the problem instance}, since this allows the second query to fully separate all receiver beliefs. 
Any other initial query is always strictly suboptimal.

We show that the \emph{optimal} adaptive querying policy can be computed via dynamic programming when $T < \infty$. 
In order to do so, we leverage the 
total ordering over receiver types.
%
Therefore, we can compute offline an optimal collection of at most $\min\{T, 2^K-1\}$ \emph{possible} queries, and thereafter use binary search to select the next query given any history of responses.
This implies a reduction from the optimal non-adaptive querying policy for $\Tilde{\cP}$ to the optimal adaptive one.

\begin{definition}[Non-Adaptive Querying Policy]
A non-adaptive querying policy with support $Q \subset \mathcal{Q}$ 
poses each of the queries in $Q$ in sequence, independent of the history of responses.
\end{definition}

\begin{restatable}{lemma}{DPna}
\label{thm:adaptive-to-nonadapative}
    Fix $K \geq 1$. 
    Let $\pi$ be the sender-optimal non-adaptive querying policy with at most $\min\{T, 2^K-1\}$ queries, and let $Q$ be its support.
    Then there exists a sender-optimal (adaptive) querying policy $\pi'$ with at most $K$ queries that only makes queries in $Q$.  Moreover, $\pi'$ can be implemented in time $O(\min\{T,2^K\})$ given access to $Q$.
\end{restatable}
\begin{proof}
    First note that if $2^K \geq T$ then the result follows trivially by taking the support of the non-adaptive querying policy to be the set of all possible queries (up to action equivalence).  So we will assume that $2^K < T$.
    
    Any (adaptive) querying policy $\pi'$ with $K$ queries can generate at most $2^K$ potential histories, each corresponding to a disjoint subinterval of receiver thresholds implied by the history of responses. 
    These subintervals are described by the at most $2^K-1$ thresholds that separate them. 
    One can therefore construct a non-adaptive querying policy $\pi$ with support $Q$ consisting of queries corresponding to each of these thresholds. 
    Querying policy $\pi$ (which makes $2^K-1$ queries) would reveal which subinterval contains the receiver's threshold, which is equivalent to the information revealed by policy $\pi'$.

    We conclude that the optimal non-adaptive policy of length $2^K-1$ is at least as informative as the optimal adaptive policy of length $K$. 
    Given the optimal non-adaptive policy $\pi$ of length $2^K-1$, its information can be simulated by an adaptive policy $\pi'$ of length $K$ via binary search: at each round, $\pi'$ selects the query from $Q$ corresponding to the midpoint threshold among all queries in $Q$ that separate types not yet excluded by the history. 
    As this policy results in a distinct subinterval of types for every possible history, it reveals which of the $2^K$ subintervals defined by $Q$ contain the receiver's threshold, and is therefore as informative as $\pi$. 
    We conclude that $\pi'$ must be optimal among all adaptive policies.
\end{proof}

\xhdr{Finite-Type Non-Adaptive Querying Policies.} Our problem now reduces to finding the best (non-adaptive) set of $\min\{T,2^K\}$ queries, given rounded type distribution $\Tilde{\cP}$.\footnote{While this is exponential in $K$, the number of queries to consider is always at most $T$ since if $2^K > T$, we can choose the set $Q$ to consist of all possible queries (of which there are at most $T-1$ up to action equivalence).\looseness-1} 
We do this via dynamic programming in~\Cref{alg:dp-na}, iteratively building solutions for larger sets of receiver beliefs.
$\sigma^*_{\Tilde{\cP}(i,j)}$ is the optimal messaging policy of~\Cref{prop:opt-binary} under prior $\Tilde{\cP}(i,j)$, where $\Tilde{\cP}(i,j)$ is the type distribution $\Tilde{\cP}$ conditioned on the type being in $\{i, \ldots, j\}$.
We use $\mathrm{ind}(q)$ to index the smallest receiver type which takes action $a=1$ in response to query $q \in \cQ$.\looseness-1

Given a set of $T$ receiver types $[T]$ where $s_1 > \cdots > s_T$, 
\Cref{alg:dp-na} keeps track of the optimal sender utility achievable with $k$ queries when in receiver subset $[i]$ for all $1 \leq i \leq T$.
By the structure induced by non-adaptivity, we can write the sender's utility for $k+1$ queries in receiver subset $[i+1]$ as a function of the optimal solution for $k$ queries in subset $[i]$.
The following result, coupled with~\Cref{thm:adaptive-to-nonadapative}, implies that we can compute the optimal adaptive querying policy in time $O(T^3)$.\looseness-1
%
    

\begin{algorithm}[t]
\caption{Computing the Optimal Non-Adaptive Querying Policy with $T$ types}\label{alg:dp-na}
\textbf{Input:} Query budget $K \in \mathbb{N}$, rounded type distribution $\Tilde{F}$
\begin{enumerate}
    \item 
    For all $1 \leq i \leq j \leq T$,
    set 
    \begin{equation*}
        V[i, j] :=
        \mathbb{E}_{(\omega,s) \sim \Tilde{F}(i,j)} \mathbb{E}_{m \sim \sigma^*_{\cP(i,j)}(\omega)} \left[ u_S(\omega, s, m) \right]
    \end{equation*}
    \\
    \textit{$\triangleright \; V[i,j]$ is the sender's expected utility from messaging optimally, given (refined) distribution over types $\Tilde{\cP}(i,j)$.}

    \item For each $j \in [T]$, 
    set $M[j,0] := V[1,j]$. 
    \item For each $k \in [K]$ and $j\in [T]$,
    compute $$M[j, k] := \max_{q \in \cQ} V[\mathrm{ind}(q)+1, j] + M[\mathrm{ind}(q), k-1].$$
    \textit{$\triangleright \; M[j,k]$ is the sender's utility from querying/messaging optimally, given joint distribution $\Tilde{F}(1,j)$ and $k$ queries.\looseness-1}
    %
    %
    \item The optimal policy then makes the $K$ queries $Q(\Tilde{F}, K)$ that obtain value $M[T, K]$.
\end{enumerate}
\textbf{Return} $Q(\Tilde{F}, K)$
\end{algorithm}

\begin{algorithm}[t]
\caption{Computing an $O(\epsilon)$-Optimal Adaptive Querying Policy with $K$ queries}\label{alg:total}
\textbf{Input:} Query budget $K \in \mathbb{N}$, 
$\epsilon>0$
\begin{enumerate}
    \item For each $\omega \in \Omega$, round down $\omega$ to the nearest quantile of $F|_{\Omega}$ that is a multiple of $\epsilon$.
    %
    \item For each $s \in \cT$, round down $s$ to the nearest quantile of $\cP$ that is a multiple of $\epsilon$. Denote the rounded joint distribution by $\Tilde{F}$.
    %
    %
    \item Run~\Cref{alg:dp-na} given $K$ and $\Tilde{F}$ as input to obtain non-adaptive querying policy $Q(\Tilde{F}, K)$.
    \item Use the adaptive-to-non-adaptive reduction of \Cref{thm:adaptive-to-nonadapative} on $Q(\Tilde{F}, K)$ to make $K$ queries adaptively.\looseness-1
\end{enumerate}
%
\end{algorithm}
\begin{restatable}{lemma}{DPnatwo}
\label{thm:DP-na}
    \Cref{alg:dp-na} computes the sender's optimal non-adaptive querying policy in $\mathcal{O}(T^{4.5} W^{2.5})$
    time.
\end{restatable}
\begin{proof}
    The algorithm begins by using Proposition~\ref{prop:opt-binary} to precompute, for each range of receiver beliefs indexed by $[i,j]$ with $i \leq j$, the value of the optimal sender's messaging policy conditional on the receiver's threshold lying in the given range. 
    These values are stored as $V[i,j]$. 
    This can be done in total time $O(T^{4.5} W^{2.5})$.

    The algorithm next computes $M[j,k]$, for $1 \leq j \leq T$ and $0 \leq k \leq K$, to be the maximum sender value achievable when the receiver's threshold is known to lie in the index range $[1,j]$ and there are $k$ queries remaining to make. 
    When $k=0$ there are no further queries, so $M[j,0] = V[1,j]$.  For $k > 0$, we determine $M[j,k]$ by enumerating all possibilities for informative queries (of which there are at most $j \leq T$). 
    As there are $TK$ total entries in $M$, each of which takes $O(T)$ time to compute, our runtime for this step is $O(T^3)$.
    (Recall that $K < T$ without loss of generality).  
\end{proof}

\xhdr{Putting It All Together.} We now have all of the pieces required to bound the performance of~\Cref{alg:total}. 
\begin{restatable}{theorem}{mainthm}
\label{thm:approx.algorithm}
    Choose any $0 < \epsilon < 1$. 
    Given a problem instance with $\Tilde{W}$ states and $\Tilde{T}$ types, one can compute a querying policy in $O(\Tilde{W} + \Tilde{T} + \epsilon^{-7})$ time using~\Cref{alg:total} and a messaging policy in $O(\epsilon^{-5})$ time using~\Cref{prop:opt-binary}, for which the sender's expected utility is at least $OPT - O(\epsilon)$, where $OPT$ is the sender's optimal expected utility at equilibrium.
\end{restatable}
\begin{proof}
    The proof follows by applying~\Cref{lem:state.discretization},~\Cref{prop:value.monotone.prob},~\Cref{thm:adaptive-to-nonadapative}, and~\Cref{thm:DP-na} with $T = O(1/\epsilon)$ and $W = O(1/\epsilon)$.
\end{proof}

\begin{corollary}\label{cor:finite}
    When $T < \infty$, one can use steps 3 and 4 of~\Cref{alg:total} to compute the optimal adaptive querying policy in time $\mathcal{O}(T^{4.5} W^{2.5})$.
\end{corollary}
\begin{proof}
    The proof of~\Cref{cor:finite} follows immediately from~\Cref{thm:adaptive-to-nonadapative} and~\Cref{thm:DP-na}.
\end{proof}

\subsection{Continuous States and Types}
\label{sec:continuous}

 To this point we have assumed for technical convenience that $\Omega$ and $\cT$ are both \emph{finite} subsets of $[0,1]$.  We next show how to relax the finiteness assumption.  To this end, suppose now that $\Omega \times \cT$ is a compact, measurable, convex subset of $[0,1]^2$ and that $F$ is a measurable atomless probability distribution over $\Omega \times \cT$ with twice-differentiable density function $f \colon \Omega \times \cT \to \mathbb{R}_{\geq 0}$. 

 One place where finiteness was used was in the statement of Assumption~\ref{ass:type2}, the increasing differences property for distribution $F$.  This assumption is used to ensure monotonicity of actions in receiver type.  For continuous types, we substitute an appropriate continuous analog: supermodularity of the density function.

\begin{assumption}\label{ass:type3}
    The distribution $F$ over pairs $(\omega,s)$ has supermodular density; i.e., it satisfies $\frac{\partial^2}{\partial \omega \partial s} f(\omega,s) \geq 0$ for all $(\omega,s) \in \Omega \times \cT$.
\end{assumption}

Given Assumption~\ref{ass:type3}, we can conclude the analog of Lemma~\ref{lem:monotone}.  This is because, for any $s_2 > s_1$, the posterior distribution over $\omega$ given $s_2$ first-order stochastically dominates the distribution over $\omega$ given $s_1$, given any message $m$ from the sender.

With Assumption~\ref{ass:type3} in place, the rounding procedure used in Algorithm~\ref{alg:total} can be applied to any $\Omega \times \cT \subseteq [0,1]^2$, resulting in a joint distribution $\tilde{F}$ with finite support.  Moreover, our error sensitivity bounds in Lemmas~\ref{lem:state.discretization} and~\ref{prop:value.monotone.prob} do not depend on finiteness of $\Omega$ and $\cT$.  Our algorithm and analysis therefore carry over without further change.

\begin{corollary}
    Choose any $0 < \epsilon < 1$ and suppose the rounding in Lemmas~\ref{lem:state.discretization} and~\ref{prop:value.monotone.prob} may be done in $O(\Gamma)$ total time, {for some $\Gamma = \Gamma(\epsilon)$}.  
    Then one can compute a querying policy in $O(\Gamma + \epsilon^{-7})$ time using~\Cref{alg:total} and a messaging policy in $O(\epsilon^{-5})$ time using~\Cref{prop:opt-binary}, {so that} the sender's expected utility is at least $OPT - O(\epsilon)$, where $OPT$ is the sender's optimal expected utility at equilibrium.
\end{corollary}

%% file: paper/main_body/extensions.tex



\subsection{Approximate Oracles}
\label{sec.approx}
Our baseline model assumes that the query oracle has perfect access to the receiver's type, which it uses to simulate the receiver's actions. 
However, our results also extend to scenarios where the oracle's access to the receiver's type is imperfect and subject to noise (as is the case when, e.g. the oracle is a machine learning model predicting the receiver's behavior). 
We show that the sender's utility will degrade smoothly with the amount of noise in the oracle.

To this end, given types $s_1, s_2 \in \cT$, write $\Delta(s_1,s_2) := |\Pr_{s\sim\cP}[s\leq s_1] - \Pr_{s\sim\cP}[s \leq s_2]|$.
We then suppose the following holds for some constants $\delta, \gamma > 0$: 
If the receiver has type $s$, then the query oracle is endowed with a type $s'$ such that $\Pr\sbr{\Delta(s,s') < \delta} > 1-\gamma$.
%
    
Then, we can consider a sender who uses an (approximately) optimal querying policy as in~\Cref{thm:approx.algorithm}, resulting in a posterior over receiver's types whose support includes $s'$.
%
The sender can then reduce each receiver type in this support 
 by $\delta$ (in probability space) 
 before constructing a messaging policy.
 With probability at least $1-\gamma$, this perturbed posterior will only under-estimate the receiver's true type, and only by at most two intervals in the discretization. 
 Thus, by Proposition~\ref{prop:value.monotone.prob}, their resulting querying policy generates at most $O(\delta)$ less utility compared to that constructed with a perfect oracle
 again with probability at least $1-\gamma$.  Since the sender's utility is unconditionally always at least $0$, we obtain the following result.

\begin{corollary}\label{cor:p}
Let $s$ be the receiver's type and suppose the query oracle simulates a receiver with type $s'$, where $\Pr\sbr{\Delta(s,s')>\delta}<\gamma$ for some $\delta,\gamma>0$.
%
%
Then we can compute querying and messaging policies for the sender that obtain expected payoff $O((1-\gamma)OPT - \delta)$, where $OPT$ is the expected payoff of the optimal policies given access to an oracle for which $s'=s$ with probability $1$.\looseness-1
\end{corollary}
%


%


%% file: paper/main_body/non_binary.tex
\section{Extensions}\label{sec:general}

We now consider several extensions of our model: one to a more general query structure called a \emph{partition query} (\Cref{sec:partition}) and one to a setting with non-binary actions (\Cref{sec:non-binary}). 
In this section, we focus on the special case where the receiver's utility function is known to the sender. 
In~\Cref{sec:partition} we show that finding the optimal querying policies is NP-Complete with partition queries, and in~\Cref{sec:non-binary} we show that in settings with $3+$ actions, simulation queries induce a more general partition over the receiver type space than thresholds. This leads to challenges when designing algorithms for determining the sender's optimal querying policy with non-binary actions.

\subsection{Partition Queries}\label{sec:partition}
The key idea behind~\Cref{alg:dp-na} is that there is always a ``total ordering'' over both receiver beliefs and (simulation) queries, and thus one can use dynamic programming to iteratively construct an optimal non-adaptive querying policy for finite types. 
But simulations, while well-motivated, are a limited type of query.
More generally, an oracle might be able to provide information about subsets of receiver types. 
In the most general query model, the sender presents a partition of the receiver type space and the oracle returns the piece of the partition in which the (true) receiver type lies.\looseness-1

\begin{definition}\label{def:query}[Partition Oracle]
    A partition oracle is characterized by the \emph{query space} $\cQ$: the set of allowable queries $Q=\{q_1,\ldots,q_k\}\in \cQ$ that  partition of receiver's beliefs
    (\ie $q_i\subset[0,1]$ with $\cup_i q_i=[0,1]$).
    %
    The oracle inputs a query $Q \in \cQ$ and returns the subset $q\in Q$ containing the receiver's type.\looseness-1
\end{definition}

In general, such \emph{partition queries} do not admit a total ordering over types and queries, so our dynamic program does not extend. In fact, we show the corresponding decision problem is NP-Complete. The decision problem is as follows. \emph{We are given: joint distribution $F$ over $\Omega \times \cT$, query space $\cQ$, $K \in \mathbb{N}$, and $u >0$.
Does there exist a querying policy $\pi$ such that, the sender achieves expected utility at least $u$ after $K$ rounds of interaction with $\pi$?}\looseness-1

\begin{restatable}{theorem}{hardness}
\label{thm:hardness}
    Finding the optimal querying policy is NP-Complete with partition queries.\looseness-1
\end{restatable}
\begin{proof}[Proof Sketch]
To prove NP-Hardness, we reduce from Set Cover. Given a universe of elements $U$, a collection of subsets $S$, and a number $K$, Set Cover 
asks if there exists a collection of subsets $S' \subseteq S$ such that $|S'| \leq K$ and $\bigcup_{\zeta \in S'} \zeta = U$. Our reduction proceeds by creating a receiver type for every element in $U$ and a partition query for every subset in $S$. We define $u$ and $\cT$ so that the sender can only achieve expected utility $u$ if she can distinguish between every pair of receivers; \ie only if after executing policy $\pi$, the sender knows the receiver's type exactly. Finally, we show that under this construction the answer to 
Set Cover 
is $\texttt{yes}$ if and only if the answer to our decision problem is also $\texttt{yes}$.\looseness-1
\end{proof}
The following definitions will be useful for the proof of~\Cref{thm:hardness}. 
\begin{definition}[Type Partition]\label{def:partition}
    A querying policy $\pi$ induces a partition $\Gamma_{\pi}$ over receiver type space such that $\bigcup_{\eta \in \Gamma_{\pi}} \eta = \cT$. 
    Receiver type $s$ belongs to the subset $\eta_{\pi}[s] \in \Gamma_{\pi}$ of all receiver types consistent with the history generated by $\pi$ for belief $s$.
\end{definition}
\begin{definition}[Complete Separation]\label{def:sep}
    We say that a querying policy $\pi$ completely separates the set of receiver types $\cT$ if, for every receiver type $s \in \cT$,
    \begin{equation*}
        \eta_{\pi}[s] \cap \cT = \{s\},
    \end{equation*}
    where $\eta_{\pi}[s]$ is defined as in~\Cref{def:partition}.
\end{definition}

%
%
%
%
We use the shorthand $\scdp(U,S,K)$ and $\nadp(\Omega, \cT, F, \cQ, K, u)$ to refer to the Set Cover and Query decision problems respectively.
\begin{proof}
    Observe that given a candidate solution $\pi$ and the set of corresponding BIC signaling policies for each receiver subset, we can check whether the sender's expected utility is at least $u$ in polynomial time, by computing the expectation. 
    This establishes that the problem is in NP. 
    To prove NP-Hardness, we proceed via a reduction from Set Cover.
    %
    %
    Given an arbitrary Set Cover decision problem $\scdp(U,S,K)$, 
    \begin{enumerate}
        \item Let $\Omega = \{0, 1\}$.
        \item Create a set of receiver beliefs $\cT(U)$. 
        Specifically, add type $s_{\varnothing}$ to $\cT(U)$, and add a type $s_e$ to $\cT(U)$ for every element $e \in U$, where these types will be specified below.
        \item For each subset $\zeta \in S$, create a query $q_{\zeta}$ that separates the receiver types $\{s_e\}_{e \in \zeta}$ from each other and from all other types $\cT(U) \backslash \{s_e\}_{e \in \zeta}$. 
        For example if $\zeta = \{1, 2, 3\}$, then $q_{\zeta} = \{ \{s_1\}, \{s_2\}, \{s_3\}, \cT(U) \backslash \{s_1, s_2, s_3\}\}$. 
        Denote the resulting set of queries by $\cQ(S)$.
        Note that each query in $\cQ(S)$ has a unique non-singleton set in the partition it induces.
        \item Let $\cP$ be the uniform prior over $\cT(U)$. 
        Set the receiver types such that each receiver type has a different optimal signaling policy.\footnote{Note that it is always possible to do this, as when $\Omega = \{0, 1\}$ the optimal signaling policy will be different for two receivers with different beliefs over $\Omega$.}
        Set $u = \mathbb{E}_{(\omega, s) \sim F} \mathbb{E}_{m \sim \sigma_{p}^*(\omega)}[u_S(\omega, s, a)]$, where $\sigma_{s}^*$ is the optimal signaling policy when the receiver is known to have type $s$. 
    \end{enumerate}
    \paragraph{Part 1:} Suppose $\scdp(U, S, K) = \texttt{yes}$. Let $S'$ denote the set of subsets that covers $S$, and let $\pi$ denote the querying policy that poses queries $\{q_{\zeta}\}_{\zeta \in S'}$, in sequence, regardless of the history of responses. 
    %
    Let us consider each $s \in \cT$ on a case-by-case basis. 
    \paragraph{Case 1.1: $s \in \cT \backslash \{s_{\varnothing}\}$.} 
    Since $S'$ covers $S$, $\{s\}$ is a partition induced by at least one query $q$ made by $\pi$ (by construction), and so $\eta_{\pi}[s] = \{s\}$.
    \paragraph{Case 1.2: $s = s_{\varnothing}$.} 
    Likewise, since $S'$ covers $S$, for every $s \in \cT \backslash \{s_{\varnothing}\}$ there is at least one query $q$ for which $\{s\}$ is a partition induced by $q$.  Therefore $p_{\varnothing}$ and $s$ are separated by $q$, which implies
    %
    $s \not\in \eta_{\pi}[s_{\varnothing}]$. We conclude that $\eta_{\pi}[s_{\varnothing}] = \{s_{\varnothing}\}$.

    Putting the two cases together, we see that $\pi$ completely separates $\cT(U)$ according to~\Cref{def:sep}, and so the sender will be able to determine the receiver's type and achieve optimal utility. 
    Therefore $\nadp(\Omega, \cT(U), F, \cQ(S), K, u) = \texttt{yes}$.
    \paragraph{Part 2:} Suppose $\scdp(U, S, K) = \texttt{no}$. 
    Recall that, by construction, for each query $q \in \cQ(S)$ there is exactly one response corresponding to a non-singleton set.
    %
    Fix any querying policy $\pi$ and 
    consider the (unique) history $H$ of queries that is generated by $\pi$ when the response from each query is its unique non-singleton set.  Let $\cQ'$ be the set of $K$ queries posed to the oracle in history $H$.
    Note that there is a one-to-one mapping between queries in $\cQ(S)$ and subsets in $S$, and so we can denote the set of subsets corresponding to $\cQ'$ by $S' := \{\zeta\}_{q_{\zeta} \in \cQ'}$. 
    Since $S'$ does not cover $S$, there must be at least one element $e_{S'} \in S \backslash (\bigcup_{z \in S'} z)$. 
    If there are multiple such elements, pick one arbitrarily. 
    Then if the receiver has type $s_{e_{S'}}$, querying policy $\pi$ will generate history $H$.  Moreover,
    by the construction of $\cQ$, we know that $s_{e_{S'}}$ falls in the same partition as $s_{\varnothing}$ 
    for every $q \in \cQ'$. 
    Therefore $\eta_{\pi}[s_{\varnothing}] \neq \{s_{\varnothing}\}$, and so $\pi$ does not completely separate $\cT(U)$ according to~\Cref{def:sep}.
    Since the sender cannot perfectly distinguish between all receiver types and our choice of $\cQ'$ was arbitrary, this implies that $\nadp(\Omega, \cT(U), F, \cQ(S), K, u) = \texttt{no}$. 
\end{proof}

\subsection{Beyond Binary Actions}
\label{sec:non-binary}
We now highlight the differences which come when considering more general action spaces than the binary setting of~\Cref{sec:equilcomp}. 
We believe that designing optimal querying algorithms for this setting is an interesting (but challenging) avenue for future research.

Suppose there are $|\cA|>2$ actions and every receiver type shares the same utility function $u_R(\omega, a)$.  By making a simulation query $q = (\sigma_q, m_q)$, the sender is specifying a convex polytope $\cR_q \subseteq \Delta^W$ for which we have a \emph{separation oracle}, a concept from optimization which can be used to describe a convex set. 
In particular, given a point $\vx \in \mathbb{R}^d$, a separation oracle for a convex body $\cK \subseteq \mathbb{R}^d$ will either (1) assert that $\vx \in \cK$ or (2) return a hyperplane $\boldsymbol{\theta} \in \mathbb{R}^d$ which separates $\vx$ from $\cK$, i.e. $\boldsymbol{\theta}$ is such that $\langle \boldsymbol{\theta}, \mathbf{y} \rangle > \langle \boldsymbol{\theta}, \vx \rangle$ for all $\mathbf{y} \in \cK$. 
Formally, we have the following equivalent characterization of the oracle's response to a simulation query:\looseness-1
\begin{restatable}{proposition}{sepfact}\label{fact:sep}
    Let $\alpha(\vp) := \sum_{\omega \in \Omega} (u_R(\omega, m_q) - u_R(\omega, a)) \cdot \sigma_q(m_q |\omega) \cdot \vp[\omega]$. 
    By making a simulation query $q = (\sigma_q, m_q)$, the sender is specifying a polytope 
    \begin{equation*}
        \cR_q := \{ \vp \in \Delta^{W} \; : \; \alpha(\vp) \geq 0, \forall a \in \cA\}
    \end{equation*}
    and the oracle returns either (i) $\vp_{\tau^*} \in \cR_q$ or (ii) $\vp_{\tau^*} \not\in \cR_q$ and the hyperplane $\alpha(\vp' - \vp_{\tau^*}) > 0$ for some $a' \in \cA$ and all $\vp' \in \cR_q$.
\end{restatable}

This is a natural extension of the threshold characterization of simulation queries in~\Cref{sec:think}. 
However, in this more general setting it may be possible to distinguish between three or more beliefs using a \emph{single} simulation query. 
Hence the natural generalization of~\Cref{alg:dp-na} is no longer polynomial time. 
The following is an example of a setting in which it is possible to distinguish between up to $d$ different beliefs using a single simulation query. 

\begin{example}\label{ex:sep-many}
    Suppose that there are $d$ states and $d$ actions, where $u_R(\omega, a) = \mathbf{1}\{\omega = a\}$. 
    Consider $d$ receiver beliefs $\vp_1, \ldots, \vp_d \in \Delta^d$ and let $\vp_i[\omega_i] = \frac{2}{d+1}$, $\forall i \in [d]$ and $\vp_i[\omega_j] = \frac{1}{d+1}$ for $j \neq i$.
    Under this setting, receiver type $i$ will take action $a_i$ when $m = a_1$ if for all $j \neq 1$,
    $\sigma(a_1| \omega_1) \cdot \frac{2}{d+1} \geq \sigma(a_1| \omega_j) \cdot \frac{1}{d+1}$.

    Therefore, a receiver with belief $\vp_1$ will take action $a_1$ when $m = a_i$ if for all $j \neq 1$, 
    $\sigma(a_1| \omega_j) = 2 \sigma(a_1|\omega_1)$. 
    Now let us consider another receiver type $i \neq 1$. 
    Type $i$ will default to taking action $i$ under this messaging policy whenever $m=a_1$, since it maximizes her expected utility over $\vp_i$, conditioned on receiving message $m = a_1$.\looseness-1
\end{example}

%% file: paper/main_body/conc.tex
\section{Conclusions and Future Research}\label{sec:conc}
We initiate the study of information design with an oracle, motivated by settings such as experimentation on a small number of users. We study a setting in which the sender in a Bayesian persuasion problem can interact with an oracle 
before trying to persuade the receiver.
We characterize the resulting equilibrium of the game, and provide an algorithm for computing $O(\epsilon)$-optimal querying and messaging polices. 
Our algorithm runs in time polynomial in $1/\epsilon$. 
Extensions to imperfect oracles, more general query structures, and costly queries are also considered.\looseness-1 

There are several exciting directions for future research. First, at a high level, it would be interesting to explore other information design settings in which players have oracle access.  While we focus on Bayesian persuasion for its simplicity and ubiquity, it would be natural to ask these same algorithmic design questions in settings such as cheap talk \citep{crawford1982strategic} or verifiable disclosure \citep{grossman1980disclosure}. 
Next, throughout this work, we assume that the oracle simulates receivers and note that relaxing this to arbitrary partition oracles results in an NP-Complete problem.  However, there might be natural subclasses of partition queries, such as interval queries or partitions into few sets, that circumvent this hardness result and are of practical interest. 
Finally, while our model extends naturally to settings with three or more actions, simulation queries become separation oracles as shown in \Cref{fact:sep}, which means that the natural generalization of our algorithms do not run in polynomial time. 
It would be interesting to design (approximation) algorithms for such settings.

%% file: paper/appendix/extension-app.tex
\onecolumn
\section{Technical Appendix}\label{app}
We say two messaging policies $\sigma$ and $\sigma'$ are {\em outcome equivalent} if for all $\omega$, $m\in\sigma(\omega)$, $m'\in\sigma'(\omega)$ and receiver prior $p$, the receiver-optimal action $a$ upon seeing $m$ equals the receiver-optimal action $a'$ upon seeing $m'$.
The following \emph{revelation principle} is well-known in the literature on persuasion with multiple receivers; we state it here for completeness.

\begin{restatable}{proposition}{rp}\label{fact:rp2}
    If $T < \infty$, for any messaging policy $\sigma$, there is an outcome-equivalent policy $\sigma'$ with just $M = T + 1$.  Moreover, these messages can be written as $\{m_0, m_1, \dotsc, m_T\}$, where a receiver of type $s_i$ will take action $a = 1$ upon receiving message $m_j$ if and only if $j \geq i$.
\end{restatable}
\begin{proof}
    
    %
    By Lemma~\ref{lem:monotone}, if a receiver of type $s$ takes action $1$ after seeing a message $m$, any receiver with type $s' \geq s$ will also take action $a=1$.
    %
    %
    %
    %
    {
    There are therefore at most $T + 1$ distinct subsets of receiver beliefs that are induced to take action $a=1$ on any message $m$ of signaling policy $\sigma$, each corresponding to a minimal type $s_i \in \mathcal{T}$ that takes the action (plus one more to denote no receiver taking the action).  
    
    Let $\sigma'$ be the messaging policy with messages $m_0, m_1, \dotsc, m_T$, such that $\sigma'(\omega) = m_i$ whenever $\sigma(\omega)$ would induce receivers with types $\{s_1, \dotsc, s_i\}$ to act (or $m_0$ if it induces no receiver to act).  Then $\sigma$ and $\sigma'$ are outcome equivalent by construction, and $\sigma'$ has the required structure of $T+1$ messages.}
\end{proof}
We note that $T+1$ messages may be necessary. For example, suppose that the state of the world is binary, $\Omega = \{0,1\}$, and there are two equally-likely receiver beliefs over the probability of the high state, $\{0.25, 0.75\}$.  In this case, a messaging policy with message space $\mathcal{M} = \{m_0, m_1, m_2\}$ that uniformly randomizes between messages $m_0$ and $m_1$ when $\omega = 0$, and uniformly randomizes between messages $m_1$ and $m_2$ when $\omega = 1$, induces unique behaviors on each message.  Indeed, any receiver that receives message $m_0$ can infer that $\omega = 0$ so they choose action $a=0$, and any receiver that receives message $m_2$ can infer that $\omega = 1$ so they choose action $a = 1$.  However, upon receiving message $m_1$, an agent with belief $p_i$ will have posterior belief $p_i$ as well, so an agent of type $p_1 = 0.75$ would take action $a=1$ upon receiving message $m_1$, whereas an agent of type $p_2 = 0.25$ would not.
\optsp*
\begin{proof}

Let $\cP$ denote the prior over types. 
Since there is a total ordering over receiver types, by~\Cref{fact:rp2} it suffices to consider policies with messages $m_0, m_L, m_{L+1}, \dotsc, m_H$, such that a receiver of type $s_i$ chooses action $a=1$ on message $m_j$ if and only if $i \leq j$. 
We claim that the sender's optimization can therefore be written as
\[\max_{\sigma} \; \sum_{i=L}^{H} \cP_i \cdot \sum_{j=i}^{H} \sum_{\omega \in \Omega} 
\Pr[\omega | s_i] \cdot \sigma(m_j | \omega) \]
\[\text{s.t. } \forall i \in [L,H], \;\; \sum_{\omega \in \Omega} 
\Pr[\omega | s_i] \cdot \sigma(m_i | \omega) \cdot u(\omega, s_i) \geq 0 \quad\text{(IC)}\]
\[\sum_{i=L}^H \sigma(m_i | \omega) \leq 1, \;\; \sigma(m_j | \omega) \geq 0 \quad \forall j \in [L, H], \; \forall \omega \in \Omega.\]
To see why, note that the objective iterates over all possible realizations of receiver type $s_i$ and corresponding posterior belief over $\omega$, and for each one we sum over all messages $m_j$ that would induce a receiver of that type to take action $1$.  The probability of receiving message $m_j$ and the state being $\omega$, given that the receiver's type is $s_i$, is then precisely $\Pr[\omega|s_i] \cdot \sigma(m_j|\omega)$.
The first constraint is incentive compatibility of the receiver types following the recommendation of the messages: a receiver with type $s_i$ will take action $a=1$ on message $m_j$ precisely if $m_j = \arg\max_{a \in \{0, 1\}} \mathbb{E}_{\omega \sim B_R(m_j, s_i)}[u_R(\omega, s_i, a)]$, and by monotonicity of the types the inequality binds only when $j=i$.
The last two constraints simply require that $\sigma$ is a well-defined messaging policy, where message $m_0$ receives all probability mass not attributed to any $m_i$ for $L \leq i \leq H$.

We next note that in an optimal solution, all IC constraints must bind with equality except possibly for message $m_H$ (and, implicitly, message $m_0$).  The reason being that if there is a message $m_i$ with $i < H$ whose IC constraint does not bind, then (a) that message is sent with positive probability, and moreover (b) there exists a state $\omega$ such that one can shift mass from $\sigma(m_i|\omega)$ to $\sigma(m_H|\omega)$ which would strictly increase the objective value.  
Also, for the maximum state of the world $\overline{\omega} = \max \Omega$, we must have $\sum_{i=L}^H \sigma(m_i|\overline{\omega}) = 1$, as otherwise we could increase $\sigma(m_H|\overline{\omega})$ and increase the objective value with no violation of constraints.  


Since the $T$ particular constraints described above are necessarily tight at every solution of the sender's optimization problem, they can be replaced with equalities and we are left with a program formulation with $W$ non-trivial linearly independent constraints.
%
The Rank Lemma therefore implies that an optimal solution to this program places weight on at most $W$ messages (see, e.g. \cite{lau2011iterative}). 
Finally, we note that this linear program can be solved in time $\tilde{O}((WT)^\omega)$, where $\omega \approx 2.37$ is the current matrix multiplication time, since that there are $O(WT)$ variables and $O(W)$ non-trivial linearly independent constraints~\cite{cohen2021solving}.
%
%
%
%
%
\end{proof}

\section{Costly Queries}
\label{app:costly.pbe}

Our results extend to the setting in which the sender can make unlimited queries, but pays a \emph{cost} for doing so. 
In order to extend to this setting with costs, we need to slightly modify the sender's utility function and the definition of the Bayesian perfect equilibrium.
Specifically, if the sender makes queries $q_1,\ldots,q_K$ (for some $K$ chosen by the sender) and the receiver takes action $a$, the sender's utility is reduced by $\sum_ic_{q_i}$ where $c_{q_i}<1$ is the cost of query $q_i$.  For any history $H$ of queries and responses observed by the sender, we can write $c(H)$ for the sum of costs of the queries posed in history $H$. The definition of Bayesian perfect equilibrium must change slightly to reflect this utility function.
This adds a cost term to the equilibrium condition for query policies.  For completeness, we describe the modified equilibrium definition in Appendix~\ref{app:costly.pbe}.\looseness-1

Our reduction from adaptive to non-adaptive querying policies does not directly extend to the costly setting. 
For example, if the median threshold in the optimal non-adaptive querying policy corresponds to a query with very high cost, it might be suboptimal to make that query first; one might instead begin with a less-balanced but cheaper query and only make the expensive query later in the sequence if necessary. 
Nevertheless, an algorithm similar to \Cref{alg:dp-na} may be used to compute an $O(\epsilon)$-optimal adaptive querying policy for the costly setting (Algorithm~\ref{alg:dp-a-cost}). 
The primary change relative to Algorithm~\ref{alg:dp-na} is that it will maintain a table with entries $M[i,j]$, rather than $M[j,k]$, where $M[i,j]$ denotes the payoff of the sender-optimal querying policy starting from the information that the sender's type lies between $s_i$ and $s_j$. 
As before, the update step for $M[i,j]$ includes an option to take the maximizer to be $V[i,j]$, which corresponds to the choice to terminate the sequence of queries.\looseness-1
\begin{corollary}
    Choose any $0 < \epsilon < 1$. 
    One can compute a querying policy in $O(\epsilon^{-7})$ time using~\Cref{alg:dp-a-cost} and a messaging policy in $O(\epsilon^{-5})$ time using~\Cref{prop:opt-binary}, for which the sender's expected utility is at least $OPT - O(\epsilon)$, where $OPT$ is the sender's optimal expected utility at equilibrium.
\end{corollary}

Below we formally define perfect Bayesian equilibrium in this modified model.

\begin{definition}
A strategy profile $(\pi^*, \sigma^*, a^*)$ and belief rules $B_S$, $B_R$
form a Perfect Bayesian equilibrium if
\end{definition}
\begin{enumerate}
    \setlength{\itemsep}{0pt}
    \setlength{\parsep}{0pt}
    \setlength{\topsep }{0pt}
    \item For each $m \in \mathcal{M}$ and $s \in \cT$, action $a^*(m, s)$ maximizes the receiver's expected utility given belief $B_R(m, s)$, i.e.
    $a^*(m, s) \in \arg\max_a \,
    \mathbb{E}_{\omega \sim B_R(m, s)}\sbr{ u_R( \omega, s, a) }$.
    \item Belief $B_R(m,s)$ is the correct posterior over $\omega$ given $s$, $\sigma^*_H$, and the fact that $\sigma^*_H(\omega) = m$.
    \item For each $H \in \mathcal{H}$, messaging policy $\sigma^*_H$ maximizes the sender's expected utility given belief $B_S(H)$, i.e.
    $\sigma^*_H \in \arg\max_{\sigma}\, 
    \mathbb{E}_{s \sim B_S(H),\; \omega \sim F|_s}\sbr{ u_S( \omega, a^*(\sigma(\omega), s)) }$.
    \item Belief $B_S(H)$ is a correct posterior over $s$, given $\pi^*$ and that $\pi^*$ generates history $H$.
    \item Sender's querying policy $\pi^*$ maximizes the sender's expected utility given $\sigma^*$ and $a^*$, i.e.
    $\pi^* \in \arg\max_{\pi}\;
    \mathbb{E}_{(\omega,s) \sim F,\; H \sim \pi }\sbr{ u_S( \omega, a^*(\sigma^*_H(\omega), s) ) - c(H) } $.
\end{enumerate}
An algorithm similar to \Cref{alg:dp-na} may be used to compute the optimal non-adaptive querying policy for the costly setting by (1) setting $K=T-1$, and (2) 
using the following modified update step: $M[j, k] := \max \{ V[1,j], \max_{q \in \cQ} V[q+1, j] + M[q, k-1] - c_q \}$, 
where taking $V[1,j]$ as the maximizer for $M[j,k]$ corresponds to terminating the sequence of queries.
\begin{corollary}
    When $T < \infty$, the above modification to~\Cref{alg:dp-na} computes the sender's optimal non-adaptive querying policy under costly queries in $\mathcal{O}(T^{4.5} W^{2.5})$ time. 
\end{corollary}

\begin{algorithm}[t]
\caption{Computing an $O(\epsilon)$-Optimal Adaptive Querying Policy with Costly Queries}\label{alg:dp-a-cost}
\textbf{Input:} Query costs $c_q \geq 0$\\
\begin{enumerate}
    \item For each $\omega \in \Omega$, round down $\omega$ to the nearest quantile of $F|_{\Omega}$ that is a multiple of $\epsilon$.
    %
    \item For each $s \in \cT$, round down $s$ to the nearest quantile of $\cP$ that is a multiple of $\epsilon$. Denote the rounded joint distribution by $\Tilde{F}$.
    \item For all $1 \leq i \leq j \leq T$,
    set 
    \begin{equation*}
        V[i, j] := 
        \mathbb{E}_{(\omega,s) \sim \Tilde{F}(i,j)} \mathbb{E}_{m \sim \sigma^*_{\cP(i,j)}(\omega)} \left[ u_S(\omega, s, m) \right],
    \end{equation*}
    %
    %
    where $\sigma^*_{\cP(i,j)}$ is the optimal messaging policy of~\Cref{prop:opt-binary} under joint distribution $\Tilde{F}(i,j)$. 
    \item For every $1 \leq i \leq T$, set
        $M[i,i] := V[i,i]$
    \item For every $1 \leq i < j \leq T$, compute 
    \begin{equation*}
        M[i, j] := \max \left\{ V[i,j], \max_{q \in \cQ} M[i,q] + M[q+1, j] - c_q \right\}
    \end{equation*}
    %
    %
    \item The policy makes the sequence of queries that obtain value $M[1, T]$
\end{enumerate}
\end{algorithm}

%

%